\documentclass[11pt]{article}

\usepackage{authblk}
\usepackage{homin}
\usepackage{homin-thm}
\usepackage{multirow}
\usepackage{enumerate}
\usepackage{epsfig}
\usepackage{amsmath, dsfont, amscd, latexsym,  color}
\usepackage{amssymb}
\usepackage{setspace}

\usepackage{ifthen}
\usepackage{color}

\newcommand{\docversion}{full}

\DeclareRobustCommand{\versions}[2]{%
    \ifthenelse{\equal{\docversion}{proceedings}}{#1}{}%
    \ifthenelse{\equal{\docversion}{full}}{#2}{}%
    \ifthenelse{\equal{\docversion}{compare}}{\textcolor{green}{$[$}\textcolor{red}{#1}\textcolor{green}{$|$}\textcolor{blue}{#2}\textcolor{green}{$]$}}{}%
}

\DeclareRobustCommand{\onlyproceedings}[1]{%
    \ifthenelse{\equal{\docversion}{proceedings}}{#1}{}%
    \ifthenelse{\equal{\docversion}{compare}}{\textcolor{red}{#1}}{}%
}

\DeclareRobustCommand{\nowherebutproceedings}[1]{%
    \ifthenelse{\equal{\docversion}{proceedings}}{#1}{}%
}

\newcommand{\R}{\mathbb{R}}
\newcommand{\Z}{{\mathbb Z}}
\newcommand{\la}{\langle}
\newcommand{\ra}{\rangle}
\newcommand{\wh}{\widehat}

\newcommand{\sgn}{\mathrm{sign}}
\newcommand{\littlesum}{\mathop{\textstyle \sum}}

\newcommand{\NS}{\mathbb{NS}}
\newcommand{\AS}{\mathbb{AS}}

\newcommand{\bits}{\{-1,1\}}
\newcommand{\bn}{\bits^n}
\newcommand{\isafunc}{: \bn \rightarrow \bits}

\newcommand{\Ex}{\mathop{{\bf E}\/}}
\newcommand{\Var}{\operatorname{Var}}

\newcommand{\commented}{no}

\ifthenelse{\equal{\commented}{yes}}{
\newcommand{\rnote}[1]{\footnote{{\bf [[Rocco: {#1}\bf ]] }}}
\newcommand{\inote}[1]{\footnote{{\bf [[Ilias: {#1}\bf ]] }}}
\newcommand{\jnote}[1]{\footnote{{\bf [[Ragesh: {#1}\bf ]] }}}
\newcommand{\anote}[1]{\footnote{{\bf [[Andrew: {#1}\bf ]] }}}
\newcommand{\lnote}[1]{\footnote{{\bf [[Li-Yang: {#1}\bf ]] }}}
}

\ifthenelse{\equal{\commented}{no}}{
\newcommand{\rnote}[1]{}
\newcommand{\inote}[1]{}
\newcommand{\jnote}[1]{}
\newcommand{\anote}[1]{}
\newcommand{\lnote}[1]{}
}

\begin{document}

\title{On the Distribution of the Fourier Spectrum of Halfspaces}

\author[1]{Ilias Diakonikolas}
\author[2]{Ragesh Jaiswal}
\author[3]{Rocco A. Servedio}
\author[3]{Li-Yang Tan}
\author[4]{Andrew Wan}
\affil[1]{University of California Berkeley, \texttt{ilias@eecs.bekeley.edu}.}
\affil[2]{Indian Institute of Technology Delhi, \texttt{rjaiswal@cse.iitd.ac.in}.}
\affil[3]{Columbia University, \texttt{\{ras2105, liyang\}@cs.columbia.edu}.}
\affil[4]{Tsinghua University, \texttt{andrew@tsinghua.edu.cn}.}

\maketitle

\begin{abstract}
Bourgain~\cite{Bourgain:02} showed that any noise stable Boolean function $f$
can be well-approximated by a junta.
In this note we give an exponential sharpening of the parameters of
Bourgain's result under the additional assumption that $f$ is a halfspace.
\end{abstract}

\section{Introduction} \label{sec:intro}

There is a sequence of results \cite{NisanSzegedy:94,
Friedgut:98, Bourgain:02} in the theory of Boolean functions which share the
following general flavor:
if the Fourier spectrum of a Boolean function $f$ is concentrated on
low-degree coefficients, then $f$ must be close to a junta (a function
that depends only on a small number of its input variables).
Bourgain's theorem \cite{Bourgain:02} is the most recent and
strongest of these results; roughly speaking, it says that if a Boolean function $f$
has low noise sensitivity then $f$ must be close to a junta.   
See Section~\ref{sec:background} for definitions and a precise
statement of Bourgain's theorem.


The parameters in the statement of Bourgain's theorem are essentially
the best possible for general Boolean functions, in the sense that the $n$-variable Majority function 
almost (but not quite) satisfies the premise of the theorem -- its noise sensitivity is only
slightly higher than the bound required by the theorem -- but is very far 
from any junta. \ignore{\rnote{Question:  is it true that Bourgain's theorem 
is (almost) tight in the other following sense:  that there are Boolean 
functions that satisfy its noise sensitivity bound but really require 
juntas of essentially Bourgain's claimed size for $\eps$-approximation, 
and can't be approximated using significantly smaller juntas?  (Obviously 
if there are such functions they won't be halfspaces by our result...) 
If we knew of such functions, they might be worth mentioning too, as 
relevant to our result -- perhaps in an appendix.  
Are such functions known?}}It
is interesting, though, to consider whether quantitative improvements of the 
theorem are possible for restricted classes of Boolean functions; this is 
what we do in this paper, by considering the special case when $f$ is a 
halfspace.  In \cite{DiakonikolasServedio:09}
a quantitatively stronger version of an earlier ``junta
theorem'' due to Friedgut \cite{Friedgut:98} was proved for the
special case of halfspaces, and it was asked whether a similarly strengthened
version of Bourgain's theorem held for halfspaces as well.
Intuitively, any halfspace which has noise sensitivity lower than that of 
Majority should be ``quite unlike Majority'' and thus could reasonably be 
expected to depend on few variables; our result makes this intuition precise.

In this note we show that halfspaces do indeed
satisfy a junta-type theorem which is similar to
Bourgain's but with exponentially better parameters.  Our main result 
shows that if $f$ is a halfspace which (unlike the Majority function) 
satisfies a noise sensitivity bound similar to the one in Bourgain's 
original theorem, then $f$ must be close to a junta of exponentially smaller 
size than is guaranteed by the original theorem.
%
Our proof does not follow either
the approach of Bourgain or of \cite{DiakonikolasServedio:09} but instead is a case analysis based
on the value of a structural parameter
known as the ``critical index'' \cite{Servedio:07cc,DGJ+:10,OS11:chow}
of the halfspace.





\subsection{Background and Statement of Main Result.} \label{sec:background}

We view Boolean functions as mappings $f : \bn \to \bits$.  
All probabilities and expectations over $x \in \bn$
are taken with respect to the uniform distribution, unless
otherwise specified.
We say that $f,g : \bn \to \bits$ are $\eps$-close to each other
(or that $g$ is an $\eps$-approximator to $f$) if $\Pr[f(x) \neq g(x)] \leq \eps$.

A function $f : \bn \to \bits$ is said to be a ``junta on
$\mathcal{J} \subseteq [n]$'' if $f$ only depends on the coordinates
in $\mathcal{J}$. We say that $f$ is a $J$-junta,
$0 \le J \le n$, if it is a junta on some set of cardinality at most
$J$.

\begin{definition}[Noise sensitivity]
The \emph{noise sensitivity} of a Boolean function $f:\bn \to \bits$ {\em at noise rate $\eps$} is defined as
\[\NS_{\eps}(f) = \Pr_{x, y}[f(x) \neq f(y)],\]
where $x$ is uniformly distributed and $y$ is obtained
from $x$ by flipping each bit of $x$ independently with probability $\eps$.
\end{definition}

Bourgain's theorem may be stated as follows (see Theorem~4.3 of 
\cite{KhotNaor06}).  \rnote{Please, someone check carefully
that the following is a correct restatement of Theorem 4.3 of [KN06] --
I want to be sure I did this correctly}

\begin{theorem}[\cite{Bourgain:02}] \label{thm:bourgain}
Fix $f: \{-1,1\}^n \to \{-1,1\}$ and $\eps,\delta$ sufficiently 
small.\footnote{Here and throughout the paper, ``sufficiently small''
means ``in the interval $(0,c)$'' where $c>0$ is some
universal constant that we do not specify.}
If $\NS_{\eps}(f) \leq (\delta \sqrt{\eps})^{1+o(1)}$,
then $f$ is $\delta$-close to a $2^{O(1/\eps)} \cdot \poly(1/\delta)$-junta.
\end{theorem}

A \emph{halfspace}, or \emph{linear threshold function} (henceforth
simply referred to as an LTF), over $\{-1,1\}^n$ is a Boolean function
$f: \{-1,1\}^n \to \{-1,1\}$ of the
form $f(x) = \sgn(\littlesum_{i=1}^n w_ix_i - \theta)$,
where $w_1, \ldots, w_n, \theta \in \R$. The function $\sign(z)$ takes
value $1$ if $z \geq 0$ and takes value $-1$ if $z < 0$; the values
$w_1,\dots,w_n$
are the \emph{weights} of $f$ and $\theta$ is the \emph{threshold}.
LTFs have been intensively studied for decades in
many different fields such as machine learning and computational
learning theory, computational complexity, and voting theory and
the theory of social choice.  

\ignore{
They are variously known as  ``halfspaces'' or ``linear separators'' in
machine learning and computational learning theory, ``Boolean threshold
functions,'' ``(weighted) threshold gates'' and
``(Boolean) perceptrons (of order 1)'' in computational complexity, and as
``weighted majority games'' in voting theory and the theory of social choice.
}

Our main result, given below, is a strengthening
of Bourgain's theorem for halfspaces:

\begin{theorem}[Main Result]\label{thm:main}
\rnote{Shall we replace 
$``O(\delta^3\sqrt{\eps})$ with ``$O(\delta^{(2-\eps)/(1-\eps)} 
\sqrt{\eps})$''?}
Fix $f:\bn \to \bits$ to be any LTF and $\eps, \delta$ sufficiently small.
If $\NS_{\eps}(f) \leq O(\delta^{(2-\eps)/(1-\eps)} \sqrt{\eps}),$
then $f$ is $\delta$-close to an $O\left( (1/\eps^2) \cdot \log(1/\eps) 
\cdot \log(1/\delta) \right)$-junta.
\end{theorem}

Theorem~\ref{thm:main} requires a slightly stronger bound on the 
noise sensitivity in terms of $\delta$, namely as much as $\delta^{(2-\eps)/(1-\eps)}$ versus 
essentially $\delta$, but the resulting junta size bound of 
Theorem~\ref{thm:main} is exponentially smaller, both as a function of $\eps$ and of $\delta$, than the bound of Theorem~\ref{thm:bourgain}.

\ignore{
\rnote{One possibility is to add here a short ``Relation to previous work'' subsection
in which we describe the relevant [DS09] result (the sharpening of Friedgut
for halfspaces) and the contrast with our new result.  Recall that [DS09] showed that if $\AS(f) \leq \I$
then $f$ is $\eps$-close to an $I^2 \cdot \poly(1/\eps)$-junta.  So [DS09]
assumed an average sensitivity bound for halfspaces whereas we assume a noise sensitivity bound.   
Since $\AS(f)=\Theta(n \cdot \NS_{1/n}(f))$ if $f$ is a halfspace, I guess
the assumption of that earlier result -- that average sensitivity is small -- is basically
equivalent to assuming that $\NS_{1/n}$ is small.  With our new result we can assume
that $\NS_\eps$ is small for (pretty much) any $\eps$, and get a junta size bound out.  Should
we have a brief subsection about this, or not?  There's something to be said for not having this paper get too long...If we do have such a subsection, do we want to contrast
the [DS09] bound and our current bound further?}
}

\section{Preliminaries} \label{sec:prelims}

\subsection{Basic Notation.} \label{sec:defs}

For $n \in \Z_+$, we denote by $[n]$ the set $\{1, 2, \ldots, n \}.$ For $a, b, \eps \in \R_+$ we write $a  \stackrel{\eps}{\approx} b$ to indicate that
$|a-b| = O(\eps)$. Let $\mathcal{N}(\mu, \sigma^2)$ denote the Gaussian distribution with mean $\mu$ and variance $\sigma^2$. Let $\phi, \Phi$
denote the probability density function (pdf) and cumulative distribution function (cdf) respectively
of a standard Gaussian random variable $X \sim \mathcal{N}(0,1)$.

\subsection{Probabilistic Facts.} \label{sec:prob}

We require some basic probability results including the standard additive Hoeffding bound (see e.g.~\cite{DP09}):
\begin{theorem} \label{thm:chb}
Let $X_1, \ldots, X_n$ be independent random variables such that for
each $j \in [n]$, $X_j$ is supported on $[a_j, b_j]$ for some $a_j,
b_j \in \R$, $a_j \le b_j$. Let $X \ = \littlesum_{j=1}^{n} X_j$.
Then, for any $t>0$,
$\Pr \big[ |X - \E[X]| \ge t \big] \le 2 \exp \left(
-2t^2/\littlesum_{j=1}^{n} (b_j-a_j)^2   \right).$
\end{theorem}

\noindent The Berry-Ess{\'e}en theorem (see e.g.~\cite{Feller})
gives explicit error bounds for the Central
Limit Theorem:

\begin{theorem} \label{thm:be} (Berry-Ess{\'e}en)
Let $X_1, \dots, X_n$ be independent random variables satisfying
$\E[X_i] = 0$ for all $i \in [n]$, $\sqrt{\littlesum_i \E[X_i^2]} =
\sigma$, and $\littlesum_i \E[|X_i|^3] = \rho_3$.  Let $S = \littlesum_i X_i /\sigma$ and let $F$ denote the cumulative distribution
function (cdf) of $S$. Then
$\sup_x |F(x) - \Phi(x)| \leq \rho_3/\sigma^3$.
\end{theorem}

\begin{definition}
A vector $w=(w_1,\dots,w_n) \in \R^n$ is said to be \emph{$\tau$-regular} if $\max_i |w_i| \le \tau \| w\|_2$.
\end{definition}

An easy consequence of the Berry-Ess{\'e}en theorem is the following fact, which says that
a $\tau$-regular linear form behaves approximately like a Gaussian up to error $O(\tau)$:

\begin{fact} \label{fact:be}
Let $w=(w_1,\dots,w_n)$ be a $\tau$-regular vector in $\R^n$ with $\|w\|_2 = 1$.
Then for any interval $[a,b] \subseteq \R$, we have
$\Pr[\littlesum_{i=1}^n w_i x_i \in (a,b]]  \stackrel{\tau}{\approx}  \Phi(b) - \Phi(a)$.
(In fact, the hidden constant in the $\stackrel{\tau}{\approx}$ is at most $2$.)
\end{fact}

We say that two real-valued random variables $X,Y$ are $\rho$-correlated if $\E[XY]=\rho$. We will need the following generalization of Fact~\ref{fact:be} which is a corollary of the two-dimensional Berry-Ess{\'e}en theorem (see e.g. Theorem~68 in~\cite{MORS:10}).

\begin{theorem} \label{thm:2D-BE}
Let $w=(w_1,\dots,w_n)$ be a $\tau$-regular vector in $\R^n$ with $\|w\|_2 = 1$.
Let $(x, y)$ be a pair of $\rho$-correlated $n$-bit binary strings, i.e. a draw of $(x,y)$ is obtained
by drawing $x$ uniformly from $\{-1,1\}^n$ and independently for each $i$ choosing $y_i \in \{-1,1\}$ to satisfy $\E[x_i y_i]  =\rho.$
Then for any intervals $I_1 \subseteq \R$ and $I_2 \subseteq \R$ we have
$\Pr[(\littlesum_{i=1}^n w_i x_i , \littlesum_{i=1}^n w_i y_i ) \in (I_1, I_2)] \stackrel{\tau}{\approx} \Pr[(X, Y) \in (I_1, I_2)],$
where $(X, Y)$ is a pair of $\rho$-correlated standard Gaussians.
\end{theorem}

\subsection{Fourier Basics over $\{-1,1\}^n$.} \label{sec:fourier}

We consider functions $f : \bn \to \R$ (though we often focus on
Boolean-valued functions which map to $\{-1,1\}$), and we view
the inputs $x$ to $f$ as being distributed according to the uniform
distribution. The set of such functions forms a
$2^n$-dimensional inner product space with inner product given by
$\la f, g \ra = \E[f(x)g(x)]$. The set of functions $(\chi_S)_{S
\subseteq [n]}$ defined by $\chi_S(x) = \prod_{i \in S} x_i$ forms a
complete orthonormal basis for this space.  We will often simply
write $x_S$ for $\prod_{i \in S} x_i$.
Given a function $f : \bn \to \R$ we define its \emph{Fourier
coefficients} by $\wh{f}(S) \eqdef \E[f(x) x_S]$, and we have that
$f(x) = \sum_S \wh{f}(S) x_S$.


As an easy consequence of orthonormality we have \emph{Plancherel's
identity} $\la f, g \ra = \sum_S \wh{f}(S) \wh{g}(S)$, which has as
a special case \emph{Parseval's identity}, $\E[f(x)^2] = \sum_S
\wh{f}(S)^2$. From this it follows that for every $f : \bn \to
\bits$ we have $\sum_S \wh{f}(S)^2 = 1$.
It is well-known and easy to show that the noise sensitivity of $f$ can be expressed
as a function of its Fourier spectrum as follows
$\NS_{\eps}(f)  = \frac{1}{2}- \frac{1}{2} \cdot \littlesum_{S \subseteq [n]} (1 - 2 \eps)^{|S|} \cdot \wh{f}(S)^2$.

\section{Proof of Theorem~\ref{thm:main}} \label{sec:main}
Fix $\eps,\delta$ sufficiently small.  
Let $f: \bn \to \bits$ be an LTF satisfying $\NS_{\eps}(f) \leq O(\delta^{\frac{2-\epsilon}{1-\epsilon}} \cdot \sqrt{\eps})$. 
We will show that $f$ is $\delta$-close to an $O\left( (1/\eps^2) \cdot 
\log(1/\eps) \cdot \log(1/\delta) \right)$-junta.

We start by observing that for $\delta^{\frac{1}{1-\eps}} < \sqrt{\eps}$ the desired 
statement follows easily; indeed, under the assumption of the theorem
$f$ is $\delta$-close to a constant function. This is formalized in the 
following simple claim which holds 
for any Boolean function:

\begin{claim} \label{claim:small-delta}
Let $f:\bn \to \bits$ be any Boolean function and $0< \delta^{\frac{1}{1-\epsilon}} < \sqrt{\eps}$.
If $\NS_{\eps}(f) \leq \delta^{\frac{2-\epsilon}{1-\epsilon}} \cdot \sqrt{\eps}$, then
$f$ is $\delta$-close to a constant function.
\end{claim}

\begin{proof}
For any Boolean function we have  $\littlesum_{S\neq \emptyset} (1 - 2 \eps)^{|S|} \cdot \wh{f}^{2}(S) \leq (1-2\eps) \cdot \littlesum_{S\neq \emptyset} \wh{f}^{2}(S) = (1-2\eps)\cdot(1-\wh{f}^{2}(\emptyset))$ where the equality follows from Parseval's identity. Therefore, we can write
$$\NS_{\eps}(f) = \frac{1}{2} \cdot \left( 1 - \wh{f}^{2}(\emptyset) - \littlesum_{S\neq \emptyset} (1 - 2 \eps)^{|S|} \cdot \wh{f}^{2}(S) \right)
\geq \eps \cdot \left( 1 - \wh{f}^{2}(\emptyset) \right)$$ which
implies $1 - \wh{f}^{2}(\emptyset) \leq \delta^{\frac{2-\eps}{1-\eps}}/\eps^{1/2} \leq
\delta$ where the first inequality follows from the assumed upper
bound on the noise sensitivity and the second uses the assumption that
$\delta^{\frac{1}{1-\eps}} <\sqrt{\eps}$. It follows that $f$ is $\delta$-close to
$\sgn(\hat{f}(\emptyset))$ and this completes the
proof. 
\end{proof}

Using the above lemma, 
for the rest of the proof we can assume that $\delta^{\frac{1}{1-\eps}} \geq 
\sqrt{\eps}$.

Fix a weight-based representation of $f$ as $f(x)=\sign(w \cdot x - \theta)$,
where we assume, without loss of generality, that $\littlesum_i w_i^2=1$ and
$|w_i| \geq |w_{i+1}| >0$, for all $i \in [n-1]$.
For $k \in [n]$, we denote $\sigma_k \eqdef \sqrt{\littlesum_{i=k}^n w_i^2}$.
The proof proceeds by case analysis based on the value of the $\eps$-critical index
of the vector $w$, which we now define.

\begin{definition}[critical index]
We define the \emph{$\tau$-critical index $\ell(\tau)$ of a
vector $w \in \R^n$} as the smallest index $i \in [n]$ for which $|w_i| \leq \tau \cdot \sigma_i$. If
this inequality does not hold for any $i \in [n]$, we define $\ell(\tau) = \infty$.
\end{definition}

The case analysis is essentially the same as the one used in~\cite{Servedio:07cc, DGJ+:10}.
Let $\ell \eqdef \ell (\eps)$ be the $\eps$-critical index of $f$. We fix a parameter
$$L(\eps, \delta) \eqdef \Theta \left( \frac{1}{\eps^2} \cdot \log(1/\eps) \cdot \log(1/\delta) \right)$$
for an appropriately large value of the constant in the $\Theta(\cdot)$.
If $\ell=1$, then the linear form behaves like a Gaussian and must be either biased or noise sensitive.  In Lemma \ref{lem:reg-ns}, we show that such an $f$ is either $\delta$-close to constant or 
has noise sensitivity $\Omega(\delta^\frac{1}{1-\epsilon}\sqrt{\log(1/\delta)} \sqrt{\eps}).$ (See Case I below.)  If $\ell>L$, then previous results \cite{Servedio:07cc} establish that $f$ is $\delta$-close to a junta. (See Case III.) Finally, for $1<\ell<L$, we consider taking random restrictions to the variables before the critical index. If a $(1-\delta)$-fraction of these restrictions result in subfunctions which are very biased, then $f$ must be $3\delta$-close to a junta over the first $L$ variables.  Otherwise, a $\delta$-fraction of the restrictions result in regular LTFs which are not very biased, and we can apply the results from Case I to show that the noise sensitivity of $f$ must be too large to satisfy the conditions of Theorem \ref{thm:main}.  We show this in Lemma \ref{lem:ns-small-ci}, Case II.  Our requirement on the noise sensitivity in Theorem \ref{thm:main}, which is probably stronger than optimal, comes from the analysis of this case.                     

We now proceed to consider each of these three cases formally.  


\noindent \textbf{Case I:} [$\ell =1$, i.e. the vector $w$ is $\eps$-regular.]
In this case, we show that $f$ is $\delta$-close to a constant function.
The argument proceeds as follows: If $|\E[f]| <1-\delta$, we prove (Lemma~\ref{lem:reg-ns}) that
$\NS_{\eps}(f) = \Omega(\delta^{\frac{1}{1-\epsilon}} \sqrt{\log(1/\delta)} \cdot \sqrt{\eps})$ contradicting the assumption of the theorem.
Hence, $|\E[f]|  \geq 1-\delta$, i.e. $f$ is $\delta$-close to a constant.
Our main lemma in this section establishes the intuitive fact that a regular LTF that is not-too-biased towards a constant function has high noise sensitivity.

\begin{lemma} \label{lem:reg-ns}
Fix $0< \eps\leq 1/2$. Let $f:\bn \to \bits$ be an $\eps$-regular LTF $f(x) = \sgn(w\cdot x -\theta)$ that has $|\E[f]| = 1-p$.
Then we have
$$\NS_{\eps}(f)  = \Omega (  (p^{\frac{1}{1-\epsilon}} \sqrt{\log(1/p)} \cdot \sqrt{\eps}) - O(\eps).$$
\end{lemma}

Case I follows easily from the above lemma. Suppose that $p \leq \delta$. 
Then the function $f$ is $\delta$-close to a constant.
Otherwise, the lemma implies that $\NS_{\eps}(f) = \Omega(\delta^{\frac{1}{1-\eps}} 
\sqrt{\log(1/\delta)} \cdot \sqrt{\eps}) - O(\eps)$; since $\delta^{\frac{1}{1-\eps}} \geq \sqrt{\eps}$,
this is $ \Omega(\delta^{\frac{1}{1-\eps}} \sqrt{\log(1/\delta)} \cdot \sqrt{\eps})$.
This contradicts our assumed upper bound on $\NS_\eps(f)$
from the statement of the main theorem.

The proof of Lemma~\ref{lem:reg-ns} proceeds by first establishing the 
analogous statement in Gaussian space (Lemma~\ref{lem:gns-lb} below)
and  then using invariance to transfer the statement to the Boolean setting.

We start by giving a lower bound on the Gaussian noise sensitivity of 
any LTF as a function the noise rate and the threshold of the LTF.
The following lemma is classical for $\theta = 0$. We were not able to 
find an explicit reference for arbitrary $\theta$, so we give
a proof for the sake of completeness.

\begin{lemma}\label{lem:gns-lb}
Let $0 < \eps \leq 1/2$ and $\theta \in \mathbb{R}$.
Let $X$ and $Y$ be $\rho \eqdef (1 - 2 \eps)$-correlated standard Gaussians. Then,
\[\Pr[\sign(X - \theta) \neq \sign(Y - \theta)] \geq (1/\pi) \cdot \mathrm{arccos}(\rho) \cdot e^{-\frac{\theta^2}{1+\rho}}.\]
\end{lemma}

\begin{proof}
Let $X$ and $Y$ be $\rho$-correlated standard Gaussians. As is well known, $(X, Y)$ can be generated as follows
\[X = Z_1 = (Z_1, Z_2)\cdot (1, 0)^{T} \ \ \textrm{and} \ \  Y = \rho \cdot Z_1 + \sqrt{1 - \rho^2} \cdot Z_2 = (Z_1, Z_2)\cdot (\rho, \sqrt{1 - \rho^2})^{T}.\]
where $Z_1$ and $Z_2$ are {\em independent} standard Gaussians. For the random variables $X - \theta$ and $Y - \theta$ we can write
\begin{eqnarray*}
X - \theta &=& \left( Z_1 - \theta , Z_2 - \theta \cdot \sqrt{\frac{1 - \rho}{1 + \rho}} \right) \cdot (1, 0)^{T} \ \ \textrm{and}  \\
Y - \theta &=& \left( Z_1 - \theta , Z_2 - \theta \cdot \sqrt{\frac{1 - \rho}{1 + \rho}} \right)\cdot (\rho, \sqrt{1 - \rho^2})^{T}.
\end{eqnarray*}
Fix $\alpha \eqdef \sqrt{\frac{1 - \rho}{1 + \rho}}$ and consider the $2$-dimensional random vector $T = \left( -Z_2 + \alpha \theta, Z_1 - \theta \right)$. Note that $T$ is orthogonal to the vector $\left( Z_1 - \theta , Z_2 - \alpha \theta \right)$.


We now observe that
\[\Pr[\sign(X - \theta) \ne \sign(Y - \theta)] = \Pr[T \textrm{ ``splits'' vectors } (1, 0) \textrm{ and } (\rho, \sqrt{1 - \rho^2})]\]
We refer to Figure~\ref{fig:one} for the rest of the proof.
\rnote{A couple of things to change in the figure, I think:

\begin{itemize}

\item Can we get rid of all the letters from the .png file 
and then add them using the tex file
later?  I think there's a way to do this -- it's a small thing
but it will be easier to parse the picture if the letters
are in the same font as the rest of the paper.

\item Let's put a big black dot at the origin and then (in the tex file)
label it ``(0,0)''.  Might help to similarly have a dot labeled ``(1,0)''.

\item The picture is a little confusing.  I don't see why $y(r)=\gamma(r)$
necessarily?  If $\alpha=0$ (so $\rho=0$) and $r$ approaches $r_0$ from above,
then it seems the angle $y(r)$ would approach $\pi/2$ but from the picture
it seems $\gamma(r)$ would approach $r/4$.

\end{itemize}

}
\begin{figure}
\centering
\includegraphics[scale=0.5]{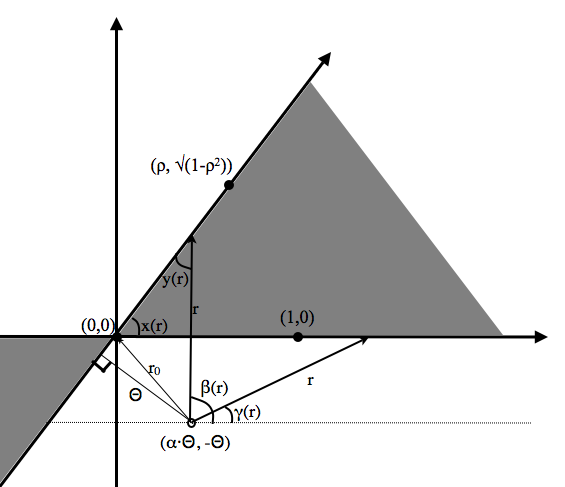} 
\caption{Illustration of the integration region for Lemma~\ref{lem:gns-lb}.}
\label{fig:one}
\end{figure}
Let $R$ be the region between the horizontal axis (the line spanned by $(1,0)$) and the line spanned by the vector $(\rho,\sqrt{1-\rho^2})$. 
The RHS of the above equation is equal to the probability mass of $R$ under a 2-dimensional unit variance Gaussian centered at $(\alpha \theta, - \theta)$.
\rnote{I'm happy with the argument up to here but then I'm confused by the
next sentence -- I don't see why the integral that's given captures
what we want.}
We estimate the Gaussian integral restricted to the region by considering points at distance $r\geq r_0$ from $(\alpha\theta,-\theta)$.  Using polar coordinates to compute the integral, we obtain:      

\begin{eqnarray}
\Pr[T \textrm{ ``splits'' vectors } (1, 0) \textrm{ and } (\rho, \sqrt{1 - \rho^2})] &\geq&    \frac{1}{\pi} \int_{r_0}^{\infty} \int_{\gamma(r)}^{\beta(r)} r e^{-r^2/2} d\phi dr \nonumber \\
 &=& \frac{1}{\pi} \int_{r_0}^{\infty} \left( \beta(r) - \gamma(r) \right) r e^{-r^2/2} dr. \label{eqn:one}
\end{eqnarray}
The angles $\beta(r), \gamma(r)$ are illustrated in Figure~\ref{fig:one}, and
$r_0$ is the distance of the point $(\alpha \theta, -\theta)$ from the origin, i.e.
\begin{equation}\label{eqn:r0}
r_0 = \theta \sqrt{1 + \alpha^2} = \frac{\sqrt{2} \theta}{\sqrt{1 + \rho}}
\end{equation}
where the second equality follows from the definition of $\alpha$.
To compute (\ref{eqn:one}), we need the following claim:
\begin{claim}\label{claim:diff}
For all $r > r_0$, it holds that $(\beta - \gamma)(r) = \mathrm{arccos}(\rho).$
\end{claim}
\begin{proof}
Let $x(r)$ and $y(r)$ denote the angles illustrated in Figure~\ref{fig:one}. First, observe that $\beta(r)= x(r)+y(r)$ and  
that $x(r) =  \mathrm{arccos}(\rho)$. We also have that $\gamma(r) = \mathrm{arcsin}(\theta/r)$ (the vector of length $r$ originates at $(\alpha \theta,-\theta)$ and stops at the origin).   
Finally, an easy calculation shows that the distance from $(\alpha\theta,-\theta)$ to the line spanned by $(\rho,\sqrt{1-\rho^2})$ is exactly $\theta$, and hence $y(r)=\mathrm{arcsin}(\theta/r)$.  
\end{proof}
Therefore, the RHS of (\ref{eqn:one}) can be written as follows:
\begin{eqnarray*}
\frac{1}{\pi} \int_{r_0}^{\infty} (\beta - \gamma)(r) r e^{-r^2/2} dr  &=& (1/\pi) \cdot \mathrm{arccos}(\rho) \cdot \int_{r_0}^{\infty} r e^{-r^2/2} dr \quad \textrm{(using Claim~\ref{claim:diff})}\\
&=& (1/\pi) \cdot \mathrm{arccos}(\rho) \left[-e^{-r^2/2} \right]_{r_0}^{\infty} \\
&=& (1/\pi) \cdot \mathrm{arccos}(\rho) \cdot e^{-r_0^2/2} \\
&=& (1/\pi) \cdot \mathrm{arccos}(\rho) \cdot e^{-\frac{\theta^2}{1 + \rho}}
\end{eqnarray*}
where the last equality follows from (\ref{eqn:r0}). This concludes the proof of Lemma~\ref{lem:gns-lb}.
\end{proof}

\noindent We are now ready to give the proof of Lemma~\ref{lem:reg-ns}.

\begin{proof}[\textbf{Proof of Lemma~\ref{lem:reg-ns}}]
We first bound from below the Gaussian sensitivity of a halfspace as a function of its bias and the noise rate.
Let $(X,Y)$ be a pair of $\rho \eqdef (1-2\eps)$-correlated standard Gaussians.
Consider the one-dimensional halfspace $h_{\theta}:\R \to \bits$ defined as $h_{\theta}(x) = \sign(x-\theta)$
and let $\left| \E_{x \sim \mathcal{N}(0,1)} [h_{\theta}(x)] \right| = 1 - \widetilde{p}$. We claim that
\begin{equation} \label{eqn:gns-bias}
\Pr[ h_{\theta}(X)  \neq h_{\theta}(Y)] = \Omega(\widetilde{p}^{\frac{1}{1-\epsilon}} \sqrt{\log(1/\widetilde{p})} \cdot \sqrt{\eps}).
\end{equation}
We show (\ref{eqn:gns-bias}) as follows: Lemma~\ref{lem:gns-lb} implies that
\begin{equation} \label{eqn:gns-one}
\Pr[ h_{\theta}(X)  \neq h_{\theta}(Y)] = \Omega (\sqrt{\eps} \cdot e^{-\frac{\theta^2}{2-2\epsilon}})
\end{equation}
where we used the elementary inequality $\mathrm{arccos}(1-2\eps) = \Omega(\sqrt{\eps})$.
We now relate $\widetilde{p}$ and $\theta$. We claim that $$\widetilde{p} =  \Theta \left( \frac{e^{-\theta^2/2}}{|\theta|+1} \right).$$
From this it follows that $$e^{-\theta^2/2} = \Theta \left( \widetilde{p} \sqrt{\log(1/\widetilde{p})} \right)$$ and
(\ref{eqn:gns-one}) yields  (\ref{eqn:gns-bias}).
It remains to get the desired bound on $\widetilde{p}$. Assume that $\theta \geq 0$; for $\theta<0$ the argument is symmetric.
First, it is easy to see that
$$ \E_{x \sim \mathcal{N}(0,1)} [h_{\theta}(x)] = -1 + 2 \widetilde{\Phi}(\theta)$$
where $\widetilde{\Phi}(\theta) \eqdef \Pr_{x \sim \mathcal{N}(0,1)} \left[ x \geq \theta \right]$. Since $\theta \geq 0$, we have $\widetilde{\Phi}(\theta) \leq 1/2$, hence
$\widetilde{p} =  2 \widetilde{\Phi}(\theta).$
The desired bound on $\widetilde{p}$ now 
follows from the following elementary fact:
\begin{fact} \label{fact:gaussian-cdf}
For all $\theta \geq 0$, it holds $\widetilde{\Phi}(\theta) = \Theta (\frac{e^{-\theta^2/2}}{|\theta|+1}).$
\end{fact}

We now turn to the Boolean setting to finish the proof of 
Lemma~\ref{lem:reg-ns}.  Let $f = \sign(w\cdot x - \theta)$ be a 
Boolean $\eps$-regular LTF (where without loss of generality $\| w\|_2=1$)
that has $|\E[f]| = 1-p$. We use~(\ref{eqn:gns-bias}) and invariance to prove 
the lemma.  The following sequence of inequalities completes the proof:
\begin{eqnarray}
\NS_{\eps}(f) &=& \Pr[\sign(w \cdot x - \theta) \neq \sign(w \cdot y - \theta)]  \nonumber \\
&\stackrel{2\eps}{\approx}& \Pr[\sign(X - \theta) \neq \sign(Y - \theta)] \label{eqn:two-dim}\\
&=& \Omega(\widetilde{p}^{\frac{1}{1-\epsilon}}\sqrt{\log(1/\widetilde{p})} \cdot \sqrt{\eps}) -O(\eps)\label{eqn:gaussian}\\
&=& \Omega (p^{\frac{1}{1-\epsilon}}\sqrt{\log(1/p)} \cdot \sqrt{\eps}) - O(\eps) \label{eqn:means}
\end{eqnarray}
where (\ref{eqn:two-dim}) follows from Theorem~\ref{thm:2D-BE} 
and (\ref{eqn:gaussian}) is an application of (\ref{eqn:gns-bias}).
To see (\ref{eqn:means}), note that, by Fact~\ref{fact:be} (a corollary of 
the Berry-Ess{\'e}en theorem) we get that $p \stackrel{\eps}{\approx} 
\widetilde{p}$, and hence $|p^{1/(1-\epsilon)} \sqrt{\log(1/p)} -  \widetilde{p}^{1/(1-\epsilon)}
\sqrt{\log(1/\widetilde{p})}| = O(\eps)$.
\end{proof}

\medskip

\noindent \textbf{Case II:} [$1< \ell \leq L$.]
In this case, we show that $f$ is $\delta$-close to an $\ell$-junta.

Consider the partition of the set $[n]$ into a set of 
head variables $H=[\ell]$ and a set of tail variables $T = [n] \setminus H$.
Let us write $H(x_H)$ to denote $w_H \cdot x_H$ and $T(x_T)$ to denote 
$w_T \cdot x_T$, the linear forms corresponding to the head and the tail.

The argument proceeds as follows: If a non-trivial fraction of restrictions to the head variables lead to a not-too-biased LTF,
we show that the original LTF has high noise sensitivity contradicting the assumption of the theorem. On the other hand,
if most restrictions to the head lead to a substantially biased LTF, we argue that the original LTF is close to a junta over the head coordinates.

Let $\rho \in \bits^{|H|}$ denote an assignment to the head coordinates and $f_{\rho}$ be the corresponding restriction of $f$.
Note that for any restriction $\rho$ of the head variables the resulting
$f_\rho$ is an $\eps$-regular LTF (with a threshold of $H(\rho)-\theta$).
Formally, we consider two sub-cases depending on the distribution of 
$|\E[f_\rho]|$ for a random choice of $\rho$.

\medskip

\noindent \textbf{Case IIa:}  [This case corresponds to $\Pr_{\rho} \big[  |\E[f_{\rho}]| \leq 1-\delta   \big] > \delta$.]
That is, at least a $\delta$ fraction of restrictions to the head variables result in a ``not-too-biased'' LTF.
Since each of these restricted sub-functions has high noise-sensitivity, we can show that the overall noise-sensitivity is also somewhat high.
This intuitive claim is quantified in the following lemma.

\begin{lemma} \label{lem:ns-small-ci}
Let $\eps,\delta$ be sufficiently small values that satisfy 
$\delta^2 \geq \sqrt{\eps}$.
Let the $\eps$-critical index $\ell$ of $f$ satisfy $1 < \ell \leq L$.
If $\Pr_{\rho} \big[  |\E[f_{\rho}]| \leq 1-\delta   \big] > \delta$,
then $\NS_{\eps}(f)  = \Omega (\delta^{\frac{2-\epsilon}{1-\epsilon}} \sqrt{\log(1/\delta)} \cdot \sqrt{\eps})$.
\end{lemma}

Therefore, in Case IIa we reach a contradiction.
To prove the above lemma, we need the following claim, whch implies
that if a noticeable fraction of restrictions to a Boolean function have high noise sensitivity, then so does the original function.

\begin{claim}\label{claim:ns-expected}
Let $f: \bn \to \bits$, $R \subseteq [n]$ and $\rho \in \bits^{|R|}$ be a random restriction to the variables in $R$. 
For any $\eps > 0$, 
if $\Pr_\rho[ \NS_\eps(f_\rho) > \tau ] > \delta$, then $\NS_\eps(f) \geq \tau \delta$. 
\end{claim}

\begin{proof}
The following elementary fact will be useful for the proof:
\begin{fact}\label{fact:restrict}
Let $f: \bn \to \bits$, $R \subseteq [n]$ and $\rho \in \bits^{|R|}$. For any $S \subseteq ([n]\setminus R)$,
$$\E_\rho[ \wh{f_\rho}(S)^2 ] = \sum_{T\subseteq R} \wh{f}(S\cup T)^2.$$
\end{fact}
\noindent By linearity of expectation and Fact~\ref{fact:restrict} we get that
\begin{equation}\label{eqn:16.1}
\E_{\rho}[\NS_{\eps}(f_{\rho})] = \frac{1}{2} \cdot \sum_{S \subseteq 
([n]\setminus R)} (1 - (1-2\eps)^{|S|}) \cdot \sum_{T \subseteq R} \hat{f}(S \cup T)^2
\end{equation}
On the other hand, we have:
\begin{eqnarray}\label{eqn:16.2}
\NS_{\eps}(f) &=& \frac{1}{2} \cdot \sum_{S \subseteq ([n]\setminus R)} 
\sum_{T \subseteq R} \left(1 - (1-2\eps)^{|S| + |T|}\right) \cdot \hat{f}(S \cup T)^2 \nonumber \\
&\geq& \frac{1}{2} \cdot \sum_{S \subseteq ([n]\setminus R)} 
\sum_{T \subseteq R} \left(1 - (1-2\eps)^{|S|}\right) \cdot \hat{f}(S \cup T)^2 \nonumber \\
&=& \frac{1}{2} \cdot \sum_{S \subseteq ([n]\setminus R)} (1 - (1-2\eps)^{|S|}) \cdot \sum_{T \subseteq R} \hat{f}(S \cup T)^2
\end{eqnarray}
Combining equations
\ref{eqn:16.1} and \ref{eqn:16.2}, we obtain 
$$\NS_\eps(f) \geq \E_\rho[\NS_\eps(f_\rho)]\geq \delta\tau.$$
\end{proof}

Using the above claim we can prove Lemma~\ref{lem:ns-small-ci}.

\begin{proof}[\textbf{Proof of Lemma~\ref{lem:ns-small-ci}}]
By Claim~\ref{claim:ns-expected} and the assumption that
$\Pr_{\rho} \big[  |\E[f_{\rho}]| \leq 1-\delta   \big] > \delta$, it suffices to show that $f_\rho$ is noise sensitive
whenever $|\E[f_\rho]| \leq 1-\delta$, i.e., that $$\NS_\eps(f_\rho)=\Omega(\delta^{\frac{1}{1-\eps}} \sqrt{\log(1/\delta)} \cdot \sqrt{\epsilon}).$$
This follows from the fact that $f_\rho$ is an $\epsilon$-regular LTF.  Applying Lemma \ref{lem:reg-ns} with $p=\delta\geq \eps^{\frac{1-\eps}{2}}$  completes the proof. 

\end{proof}

\medskip

\noindent \textbf{Case IIb:} [The complementary case corresponds to
$\Pr_{\rho} \big[ |\E[f_{\rho}]| \leq 1-\delta \big] \leq \delta$.]
That is, with probability at least $1-\delta$ over a random
restriction of the head, the bias of the corresponding restriction is
``large.'' In this case, a simple argument yields the following:

\ignore{

\begin{lemma}\label{lemma:IIb}
Fix $0< \eps \leq 1/2$ and $0 < \delta \leq 1$ satisfying $\delta^2 \geq \sqrt{\eps}$.
Let the $\eps$-critical index $\ell$ of $f$ satisfy $1 < \ell \leq L$.
If $\Pr_{\rho} \big[  |\E[f_{\rho}]| \leq 1-\delta   \big]  \leq \delta$,
 $f$ is $2\delta$-close to an $\ell$-junta.
\end{lemma}
\begin{proof}
Let $f(x) = \sign \left( H(x_H) + T(x_T) - \theta \right)$. We will show that $f$ is $2\delta$-close to the $\ell$-junta $f'(x_H) = \sign\big( H(x_H) - \theta \big)$.
Fix a good restriction $\rho$ to the head coordinates, i.e. one such that $|\E[f_{\rho}]| \geq 1-\delta$.
Then we have that either $\Pr_{x_T} [f_{\rho}(x_T) = 1] \geq 1-\delta$ or $\Pr_{x_T} [f_{\rho}(x_T) = -1] \geq 1-\delta .$
We claim that for a good restriction $\rho$, it holds $\Pr_{x_T}[ f_{\rho}(x_T) \neq f'(\rho) ] \leq \delta$.

Consider the case that $H(\rho) - \theta \geq 0$, i.e. $f'(\rho)=1$.
(An entirely similar argument works for the case $H(\rho) - \theta < 0$.)
We want to show that $\Pr_{x_T} [f_{\rho}(x_T) = 1] \geq 1-\delta$. Suppose, for the sake of contradiction, that
$\Pr_{x_T} [f_{\rho}(x_T) = -1] \geq 1-\delta .$ This means that with probability at least $1-\delta$ over the choice of $x_T$ we have
$T(x_T) < - (H(\rho) - \theta) <0$.
The regularity of the tail implies that $T(x_T)$ is nearly symmetric around $0$, hence the latter statement is a contradiction.
Indeed, by Fact~\ref{fact:be} we have that $\Pr_{x_T}[T(x_T) < 0]  \leq 1/2+ 2\eps < 1-\delta$, for $\eps, \delta$ sufficiently small positive constants.

Since at least $1-\delta$ fraction of the restrictions are good, by a union bound it follows that $\Pr [f(x) \neq f'(x)] \leq 2\delta$ as desired.
\end{proof}

}


\begin{lemma}\label{lemma:IIb}
  Let $f\isafunc$, $H \subseteq [n]$, and $0 < \delta \leq
  1$. Suppose $\Pr_{\rho\sim H}\big[|\E[f_\rho]| \leq 1-\delta\big]
  \leq \delta$. Then $f$ is $3\delta$-close to a junta over $H$. 
\end{lemma}

\begin{proof}
  Let $B\subseteq \bits^{|H|}$ denote the set of bad restrictions,
  where we say that a restriction $\rho\in \bits^{|H|}$ is bad if
  $|\E[f_\rho]| \leq 1-\delta$.  Define $g\isafunc$ to be:
\[ 
g(x) = \left\{ 
\begin{array}{cl}
1 & \text{if $x_H\in B$} \\ 
f(x)& \text{otherwise,} 
\end{array}
\right. 
\]
and note that $g$ is $\delta$-close to $f$ since $|B| \leq \delta\cdot
2^{|H|}$ by assumption.  We also have that $g$ satisfies
$|\widehat{g_\rho}(\emptyset)| = |\E[g_\rho]|
> 1 - \delta$ for all $\rho \in \bits^{|H|}$. Now consider $h(x) =
\sum_{S\subseteq H} \hat{g}(S)x_S$ and note that
\[ \|h-g\|_2^2 = \sum_{T\not\subseteq H}\hat{g}(T)^2 = \Ex_{\rho\sim
  H}[\Var(g_\rho)] = 1 - \Ex_{\rho\sim
  H}[\widehat{g_\rho}(\emptyset)^2] < 1 - (1-2\delta) = 
2\delta. \] Since $f$ is $\delta$-close to $g$ and $g$ is
$2\delta$-close to $\sgn(h)$ (a junta over $H$), this completes the
proof. 
\end{proof}

\noindent This completes Case II.

\medskip

\noindent \textbf{Case III:} [$\ell > L$].
In this case, we merely observe that $f$ is $\delta$-close to an $L$-junta.
This follows immediately from the arguments in~\cite{Servedio:07cc, DGJ+:10}. 
In particular,

\begin{lemma}[Case II(a) of Theorem~1 of \cite{Servedio:07cc}]\label{lemma:case3}
Fix $\eps, \delta>0$.
Let $f$ be an LTF with $\eps$-critical index $\ell >L$. Then $f$ is $\delta$-close to an $L$-junta.
\end{lemma}


\noindent The proof of Theorem~\ref{thm:main} is now complete.



\inote{Things to do:

\begin{itemize}

\item update the table -- some entries need reworking.  Any way to make it
visually better somehow?

\end{itemize}

}

\noindent \textbf{Acknowledgements.} I.D. would like to thank Guy Kindler and Elad Hazan for useful discussions.

\bibliography{allrefs}

\newcommand{\etalchar}[1]{$^{#1}$}
\begin{thebibliography}{MORS10}

\bibitem[Bou02]{Bourgain:02}
J.~Bourgain.
\newblock On the distributions of the fourier spectrum of boolean functions.
\newblock {\em Israel J. Math.}, 131:269--276, 2002.

\bibitem[DGJ{\etalchar{+}}10]{DGJ+:10}
I.~Diakonikolas, P.~Gopalan, R.~Jaiswal, R.~Servedio, and E.~Viola.
\newblock Bounded independence fools halfspaces.
\newblock {\em SIAM J. on Comput.}, 39(8):3441--3462, 2010.

\bibitem[DP09]{DP09}
D.~Dubhashi and A.~Panconesi.
\newblock {\em Concentration of measure for the analysis of randomized
  algorithms}.
\newblock Cambridge University Press, Cambridge, 2009.

\bibitem[DS09]{DiakonikolasServedio:09}
I.~Diakonikolas and R.~Servedio.
\newblock Improved approximation of linear threshold functions.
\newblock In {\em Proc.\ 24th Annual IEEE Conference on Computational
  Complexity (CCC)}, pages 161--172, 2009.

\bibitem[Fel68]{Feller}
W.~Feller.
\newblock {\em An introduction to probability theory and its applications}.
\newblock John Wiley \& Sons, 1968.

\bibitem[Fri98]{Friedgut:98}
E.~Friedgut.
\newblock Boolean functions with low average sensitivity depend on few
  coordinates.
\newblock {\em Combinatorica}, 18(1):474--483, 1998.

\bibitem[KN06]{KhotNaor06}
S.~Khot and A.~Naor.
\newblock {Nonembeddability theorems via Fourier analysis}.
\newblock {\em Mathematische Annalen}, 334(4):821--852, 2006.

\bibitem[MORS10]{MORS:10}
K.~Matulef, R.~O'Donnell, R.~Rubinfeld, and R.~Servedio.
\newblock Testing halfspaces.
\newblock {\em SIAM J. on Comput.}, 39(5):2004--2047, 2010.

\bibitem[NS94]{NisanSzegedy:94}
N.~Nisan and M.~Szegedy.
\newblock {On the degree of {B}oolean functions as real polynomials}.
\newblock {\em Comput. Complexity}, 4:301--313, 1994.

\bibitem[OS11]{OS11:chow}
R.~O'Donnell and R.~Servedio.
\newblock {The Chow Parameters Problem}.
\newblock {\em SIAM J. on Comput.}, 40(1):165--199, 2011.

\bibitem[Ser07]{Servedio:07cc}
R.~Servedio.
\newblock {Every linear threshold function has a low-weight approximator}.
\newblock {\em Comput. Complexity}, 16(2):180--209, 2007.

\end{thebibliography}
\bibliographystyle{alpha}

\end{document}